\documentclass[conference,letterpaper]{IEEEtran}
\IEEEoverridecommandlockouts
\usepackage{cite}
\usepackage{url}
\usepackage{ifthen}
\usepackage{algorithm, algorithmicx, algpseudocode}
\usepackage{graphicx} 
\usepackage{textcomp}
\usepackage{booktabs}
\usepackage{makecell} 
\usepackage[cmex10]{amsmath}
\usepackage{amssymb,amsfonts}
\usepackage[left=1.62cm,right=1.62cm,top=1.85cm]{geometry}
\usepackage{pifont}
\usepackage{amsthm}
\usepackage{mathrsfs}
\def\BibTeX{{\rm B\kern-.05em{\sc i\kern-.025em b}\kern-.08em T\kern-.1667em\lower.7ex\hbox{E}\kern-.125emX}}
\usepackage{tikz}
\usetikzlibrary{automata,positioning,calc}
\usetikzlibrary{intersections}
\usetikzlibrary{decorations.pathreplacing,angles,quotes}

\usepackage{bm}
\usepackage{amscd}

\algnewcommand{\Initialize}[1]{%
  \State \textbf{Initialization:}
  \Statex \hspace*{\algorithmicindent}\parbox[t]{0.8\linewidth}{\raggedright #1}
}
\makeatletter

\theoremstyle{definition}
\newtheorem{theorem}{Theorem}
\newtheorem{definition}{Definition}
\newtheorem{lemma}{Lemma}
\newtheorem{prop}{Proposition}
\newtheorem{cor}{Corollary}
\newtheorem{remark}{Remark}
\newtheorem{eg}{Example}

\newcommand{\norm}[1]{\left\lVert#1\right\rVert}

\newcommand{\argmin}{\operatornamewithlimits{argmin}}

\def\vE{\mathbb E}

\font\b=cmr10 scaled\magstep4

\def\bigzerou{\smash{\lower1.7ex\hbox{\b 0}}}
\def\bigzerou{\smash{\lower1.7ex\hbox{\b 0}}}

\begin{document}
\title{UMVUE-Type Estimators under Bregman Losses
\thanks{This work was supported by JSPS KAKENHI Grant Number 26K14712 and JST CRONOS Grant Number JPMJCS25N5.}
}

\author{
\IEEEauthorblockN{Akira Kamatsuka}
\IEEEauthorblockA{Shonan Institute of Technology \\ 
Email: \text{kamatsuka@info.shonan-it.ac.jp}
 }
\and
\IEEEauthorblockN{Shun Watanabe}
\IEEEauthorblockA{Tokyo University of Agriculture and Technology \\ 
Email: \text{shunwata@cc.tuat.ac.jp}
 } 
}
\maketitle

\begin{abstract}
We study unbiased estimation under Bregman losses and develop an extension of the classical theory of uniformly minimum variance unbiased estimators (UMVUEs). 
Exploiting bias--variance-type decompositions for Bregman divergences, we consider two natural loss functions, $D_{\varphi}(\theta,\hat{\theta})$ and $D_{\varphi}(\hat{\theta},\theta)$, and their corresponding notions of unbiasedness.
We show that the latter formulation reduces to the classical setting, whereas the former yields a different framework in which unbiasedness is characterized in the dual space induced by $\nabla\varphi$. 
For the nontrivial case, we establish analogs of the Rao--Blackwell and Lehmann--Scheff{\'e} theorems, providing a systematic construction of type-I Bregman UMVUEs. 
\end{abstract}

\section{Introduction} \label{sec:intro}

Consider the problem of estimating an unknown parameter $\theta$ from a random sample
$X^{n} = (X_{1},\dots,X_{n}) \overset{\mathrm{i.i.d.}}{\sim} p_{X\mid\theta}$.
In classical mathematical statistics, extensive attention has been devoted to the study of unbiased estimation and the construction of optimal estimators.
Fundamental results include the Cram{\' e}r--Rao lower bound \cite{Rao1950} and the Rao--Blackwell theorem \cite{10.1214/aoms/1177730497}.

In the framework of statistical decision theory established by Wald \cite{10.1214/aoms/1177730030},
an estimator $\hat{\theta}=\delta(X^{n})$ is evaluated via a loss function $\ell(\theta, \hat{\theta})$ and the corresponding risk function
$R(\theta,\delta):=\mathbb{E}_{\theta}[\ell(\theta,\delta(X^{n}))]$. 
For the point estimation of a scalar parameter, a commonly used loss is the squared-error loss
$\ell_{\mathrm{sq}}(\theta,\hat{\theta}) := (\theta-\hat{\theta})^{2}$, 
under which the risk admits the well-known bias--variance decomposition: 
\begin{align}
R(\theta,\delta)
&=
\Bigl(\mathbb{E}_{\theta}[\delta(X^{n})]-\theta\Bigr)^{2}
+
\mathrm{Var}_{\theta}(\delta(X^{n})).
\label{eq:BV_decomposition}
\end{align}
In this setting, unbiased estimators with the smallest variance among all unbiased estimators are known as uniformly minimum variance unbiased estimators (UMVUEs). The Lehmann--Scheff{\'e} theorem \cite{8539305f-dc37-3fe8-8b38-6d3a9badc255,577792d0-252f-3c76-bfd7-a52925d34054} provides a sufficient condition for an estimator to be a UMVUE.

A natural direction for extending the classical theory is to consider more general estimation settings, including loss functions beyond the squared-error loss and the estimation of vector-valued parameters.
Among such extensions, Bregman divergences \cite{BREGMAN1967200} form a broad and well-structured class that includes the squared-error loss as a special case.
Bregman losses have also been investigated in statistical estimation. For example, Dytso \emph{et al.} \cite{9625998} derived a Cram{\'e}r--Rao-type lower bound for Bayes risk under Bregman loss and investigated linear estimation. 
Banerjee \textit{et al.} \cite{1459065} showed that the conditional expectation is optimal for prediction under Bregman losses.
Extensions of the bias--variance decomposition to general loss functions, including cross-entropy and margin-based losses, have also been studied \cite{6790909,pmlr-v151-wood22a}.
Recently, bias--variance-type decompositions tailored to Bregman divergences have been derived \cite{pfau2025generalizedbiasvariancedecompositionbregman,gupta2022ensembles,heskes2026biasvariancedecompositionsexclusiveprivilege}.
In a related line of work, Das \emph{et al.} \cite{10988591} used the bias--variance-type decomposition to extend robustness results for interpolation to Bregman losses.

Despite these developments, previous studies have primarily focused on risk decomposition, prediction, or related applications, rather than on the development of an unbiased estimation theory analogous to the classical UMVUE framework.
Particularly, once such a decomposition is available, some natural questions are whether the vanishing of the corresponding bias term induces a meaningful notion of unbiasedness and whether an analog of the UMVUE can be characterized via Rao--Blackwell and Lehmann--Scheff{\'e}-type arguments.

In the Bregman loss setting, these questions depend on how the bias term is defined from the decomposition.
Indeed, the two losses
\begin{align}
\ell_{\varphi}(\theta, \hat{\theta}) &:= D_{\varphi}(\theta,\hat{\theta}), \\ 
\tilde{\ell}_{\varphi}(\theta, \hat{\theta}) &:= D_{\varphi}(\hat{\theta},\theta)
\end{align}
yield different bias-variance-type decompositions and hence different notions of unbiasedness.
As we show, the formulation based on $\tilde{\ell}_{\varphi}(\theta,\hat{\theta})$ reduces to the classical setting, whereas the formulation based on $\ell_{\varphi}(\theta,\hat{\theta})$ yields a nontrivial extension in which the bias term is no longer expressed directly as the estimator's expectation but instead as the expectation of a transformed estimator induced by the mapping $\nabla \varphi$.

Motivated by this observation, we develop a theory of UMVUEs under Bregman losses, which we call \emph{Bregman UMVUEs}.
The major contributions of this study are summarized as follows:
\begin{itemize}
\item We show that the formulation based on $\tilde{\ell}_{\varphi}(\theta,\hat{\theta})$ reduces to the classical unbiased estimation framework: the induced notion of unbiasedness coincides with the standard definition, and the corresponding optimal estimator is obtained.

\item Meanwhile, we show that the formulation based on $\ell_{\varphi}(\theta,\hat{\theta})$ yields a new notion of unbiasedness characterized by the expectation of a transformed estimator under the mapping $\nabla\varphi$ (Definition \ref{def:Bregman_unbiasedness} and Proposition \ref{prop:characterization_Bregman_unbiasedness}).

\item We derive extensions of the Rao--Blackwell and Lehmann--Scheff{\'e} theorems for this nontrivial setting, establishing a principled construction of the Bregman UMVUE via conditioning on sufficient statistics and its uniqueness under completeness conditions (Theorems \ref{thm:Bregman_Rao_Blackwell} and \ref{thm:Bregman_Lehmann_Scheffe}).
\end{itemize}

\section{Preliminaries}\label{sec:preliminaries}

Let $\mathcal{P}=\{p_{X\mid \theta}\}_{\theta\in \Theta}$ be a statistical model parameterized by $\theta\in \Theta \subseteq \mathbb{R}^{d}$. 
Let $X^{n}=(X_{1}, \dots, X_{n})\overset{\mathrm{i.i.d.}}{\sim} p_{X\mid \theta}$ be a random sample from $p_{X\mid \theta}$. 
Let $\delta(X^{n})$ be an estimator of $\theta$ based on $X^{n}$. 
Let $\vE_{\theta}[\cdot]$ denote expectation under $p_{X\mid \theta}$. 
For a scalar random variable $X$, we define 
$\mathrm{Var}_{\theta}(X):=\vE_{\theta}\bigl[(X - \vE_{\theta}[X])^{2}\bigr]$.
In this section, we review the theory of unbiased estimators and Bregman divergences.

\subsection{Theory of Unbiased Estimators}

\begin{definition}[Unbiased estimator {\cite[Chapter 2, Def.~1.1]{LC}}] \label{def:unbiased}
An estimator $\delta(X^{n})$ of $\theta$ is \textit{unbiased} if for all $\theta\in \Theta$, 
\begin{align}
\vE_{\theta}[\delta(X^{n})] = \theta.
\end{align}
\end{definition}

\begin{definition}[Sufficient statistic {\cite[p.~32]{LC}}] \label{def:sufficient_statistic}
A statistic $T:=t(X^{n})$ is \textit{sufficient} for $\mathcal{P}$ if the conditional distribution of $X^{n}$ given $T=t$ is independent of $\theta$.
\end{definition}

\begin{prop}[Rao--Blackwell Theorem {\cite[Chapter 1, Thm.~7.8]{LC}}] \label{prop:RB_thm}

Let $T$ be a sufficient statistic for $\mathcal{P}$. 
Let $\delta(X^{n})$ be an estimator with finite expectation and $\ell(\theta, \hat{\theta})$ be convex in $\hat{\theta}$. 
Define 
\begin{align}
\tilde{\delta}(T):=\vE[\delta(X^{n})\mid T]. \label{eq:RB}
\end{align}
Then,
\begin{align}
R(\theta, \tilde{\delta}) \leq R(\theta, \delta).
\end{align}
Particularly, for the point estimation of a scalar parameter $\theta$ under the squared-error loss,
if $\delta(X^{n})$ is unbiased, then $\tilde{\delta}(T)$ is unbiased and
\begin{align}
\mathrm{Var}_{\theta}(\tilde{\delta}(T)) \leq \mathrm{Var}_{\theta}(\delta(X^{n})).
\end{align}
\end{prop}

\begin{remark}
Eq.~\eqref{eq:RB} is referred to as \textit{Rao--Blackwellization} and $\tilde{\delta}(T)$ is referred to as \textit{Rao--Blackwell estimator}.
\end{remark}

\begin{definition}[UMVUE {\cite[Def.~1.6]{LC}}] \label{def:UMVUE}
For the point estimation of a scalar parameter $\theta$, an unbiased estimator $\delta(X^{n})$ is the UMVUE of $\theta$ if for all $\theta\in \Theta$ and for all unbiased estimators $\delta^{\prime}(X^{n})$,
\begin{align}
\mathrm{Var}_{\theta}(\delta(X^{n})) \leq \mathrm{Var}_{\theta}(\delta^{\prime}(X^{n})).
\end{align}
\end{definition}

\begin{definition}[Complete sufficient statistic {\cite[p.~42]{LC}}] \label{def:complete_sufficient_statistic}
A sufficient statistic $T$ is \textit{complete} if 
\begin{align}
\vE_{\theta}[h(T)] = 0 \ \text{for all } \theta\in \Theta
\;\;\Rightarrow\;\;
h(T)=0 \quad \text{a.e. } \mathcal{P}, 
\end{align}
where a.e. $\mathcal{P}$ is an abbreviation for almost everywhere with respect to every $P\in \mathcal{P}$.
\end{definition}

\begin{lemma}[{\cite[Chapter 2, Lemma~1.10]{LC}}] \label{lem:uniqueness}
Let $T$ be a complete sufficient statistic for $\mathcal{P}$. Let $\tilde{\delta}_{1}(T)$ and $\tilde{\delta}_{2}(T)$ be two unbiased estimators of $\theta$.
Then,
\begin{align}
\tilde{\delta}_{1}(T) = \tilde{\delta}_{2}(T), \quad \text{a.e. } \mathcal{P}.
\end{align}
\end{lemma}

\begin{prop}[Lehmann--Scheff{\' e} Theorem {\cite[Chapter 2, Thm.~1.11]{LC}}] \label{prop:LS_thm}
Let $T$ be a complete sufficient statistic for $\mathcal{P}$, and let $\ell(\theta, \hat{\theta})$ be convex in $\hat{\theta}$. 
If $\delta(X^{n})$ is an unbiased estimator of $\theta$, then the Rao--Blackwell estimator $\tilde{\delta}(T):=\vE[\delta(X^{n}) | T]$ uniformly minimizes the risk, 
i.e., for all $\theta\in \Theta$ and for all unbiased estimators $\delta^{\prime}(X^{n})$, 
\begin{align}
R(\theta, \tilde{\delta})\leq R(\theta, \delta^{\prime}).
\end{align}
For the point estimation of a scalar parameter with squared-error loss, $\tilde{\delta}(T)$ is the unique UMVUE of $\theta$.
\end{prop}

\begin{cor} \label{cor:LS_thm}
Let $T$ be a complete sufficient statistic for $\mathcal{P}$.
For the point estimation of a scalar parameter with squared-error loss, any unbiased estimator that is a function of $T$ is the unique UMVUE.
\end{cor}

\subsection{Bregman Divergence and Bias--Variance Decomposition} \label{ssec:Bregman_divergence}

\begin{definition}[Bregman divergence {\cite{BREGMAN1967200}}]
Let $\Omega\subseteq \mathbb{R}^{d}$ be a convex set and $\varphi:\Omega\to\mathbb{R}$ be strictly convex and differentiable.
The \textit{Bregman divergence} is defined by
\begin{align}
D_{\varphi}(x,y)
:=
\varphi(x)-\varphi(y)-\langle \nabla \varphi(y), x-y \rangle.
\end{align}
\end{definition}

\begin{eg}
Examples of Bregman divergences include:
\begin{itemize}
\item If $\varphi(x) = \frac{1}{2}\norm{x}_{2}^{2}$, then $D_{\varphi}(x, y) = \frac{1}{2}\norm{x-y}_{2}^{2}$ (half the squared Euclidean distance). 
\item If $\varphi(x) = \frac{1}{2}x^{\top}Ax$, then $D_{\varphi}(x, y)=\frac{1}{2}(x-y)^{\top}A(x-y)$ (Mahalanobis distance). 
\item If $\varphi(x) = \sum_{i=1}^{d} ( x_{i}\log x_{i} - x_{i} )$, then $D_{\varphi}(x, y) = \sum_{i=1}^{d} \left( x_{i}\log \frac{x_{i}}{y_{i}} - x_{i} + y_{i} \right)$ (generalized KL divergence). 
\item If $\varphi(x) = -\sum_{i=1}^{d}\log x_{i}$ on $\mathbb{R}_{++}^{d}$, then $D_{\varphi}(x, y) = \sum_{i=1}^{d} \left( \frac{x_{i}}{y_{i}} - \log \frac{x_{i}}{y_{i}} - 1 \right)$ (Itakura--Saito divergence \cite{1570854175842518528}). 
\end{itemize}
\end{eg}

\begin{remark}
From this point onward, we assume that $\varphi$ is closed, proper, and strictly convex. 
Under this assumption, $\varphi$ is a Legendre-type function \cite[Section~26]{rockafellar-1970a}, \cite[Section~2.3]{bauschke1997legendre}, 
which ensures that the gradient mapping $\nabla \varphi$ is bijective and satisfies 
\begin{align}
\nabla \varphi^{*} = (\nabla \varphi)^{-1}, \label{eq:nabla_formula}
\end{align}
where $\varphi^{*}$ denotes the convex conjugate of $\varphi$, defined by
\begin{align}
\varphi^{*}(x^{*}) := \sup_{x\in \Omega} \left\{ \langle x^{*}, x\rangle - \varphi(x) \right\}.
\end{align}
In the context of the mirror descent method, the map $\nabla \varphi$ (or $\varphi$ itself) is often referred to as a \textit{mirror map}.
\end{remark}

\begin{prop}[\text{\cite[Sections 2.1 and 2.2]{10.5555/1283383.1283463}}] \label{prop:basic_property_Bregman_divergence}
The following hold:
\begin{enumerate}
\item $D_{\varphi}(x, y)\geq 0$, with $D_{\varphi}(x, y) = 0$ if and only if $x=y$.
\item $x\mapsto D_{\varphi}(x, y)$ is strictly convex.
\item $D_{\varphi}(x, y) = D_{\varphi^{*}}(\nabla \varphi(y), \nabla \varphi(x))$.
\end{enumerate}
\end{prop}

\begin{prop}[\text{\cite[Thm.~2.3]{pfau2025generalizedbiasvariancedecompositionbregman}}] \label{prop:BV_decomposition_Bregman}
Let $X$ be a random variable on $\Omega$. Then, the expected Bregman divergences admit the following decompositions:
\begin{align}
\vE[D_{\varphi}(x, X)] &= D_{\varphi}(x, x^{*}) + \vE[D_{\varphi}(x^{*}, X)], \\
\vE[D_{\varphi}(X, y)] &= D_{\varphi}(\vE[X], y) + \vE[D_{\varphi}(X, \vE[X])],
\end{align}
where
\begin{align}
x^{*} = \argmin_{x\in \Omega} \vE[D_{\varphi}(x, X)].
\end{align}
\end{prop}

\begin{remark}
It follows from \cite[Lemma~2.2]{pfau2025generalizedbiasvariancedecompositionbregman} 
that the minimizer $x^{*}$ is characterized by
\begin{align}
\nabla \varphi(x^{*}) = \vE[\nabla\varphi (X)]. \label{eq:characterization_x_star}
\end{align}
\end{remark}

\section{Bregman UMVUEs} \label{sec:Bregman_UMVUE}
In this section, we formulate the problem of finding Bregman UMVUEs based on the two Bregman losses 
$\ell_{\varphi}(\theta, \hat{\theta}) := D_{\varphi}(\theta,\hat{\theta})$, $\tilde{\ell}_{\varphi}(\theta, \hat{\theta}) := D_{\varphi}(\hat{\theta},\theta)$, 
and their corresponding risk functions $R_{\ell_{\varphi}}(\theta, \delta):=\vE_{\theta}[\ell_{\varphi}(\theta, \delta(X^{n}))]$ and $R_{\tilde{\ell}_{\varphi}}(\theta, \delta):=\vE_{\theta}[\tilde{\ell}_{\varphi}(\theta, \delta(X^{n}))]$, respectively.
From the bias--variance decomposition in Proposition \ref{prop:BV_decomposition_Bregman}, these risk functions can be decomposed as follows.

\begin{prop} \label{prop:BV_decomposition_risk}
The following identities hold:
\begin{align}
R_{\ell_{\varphi}}(\theta, \delta)
&=D_{\varphi}(\theta, \hat{\theta}^{*}_{\delta}) + \vE_{\theta}\left[D_{\varphi}(\hat{\theta}^{*}_{\delta}, \delta(X^{n}))\right], \label{eq:BV_decomposition_type_I}
\\
R_{\tilde{\ell}_{\varphi}}(\theta, \delta)
&=D_{\varphi}(\vE_{\theta}[\delta(X^{n})], \theta)+\vE_{\theta}\left[D_{\varphi}(\delta(X^{n}), \vE_{\theta}[\delta(X^{n})])\right], \label{eq:BV_decomposition_type_II}
\end{align}
where
\begin{align}
\hat{\theta}^{*}_{\delta} &:= \argmin_{\hat{\theta}}\vE_{\theta}\left[D_{\varphi}(\hat{\theta}, \delta(X^{n}))\right].
\end{align}
\end{prop}

\begin{remark}
For the point estimation of a scalar parameter, when $\varphi(x)=\frac{1}{2}x^{2}$, these two decompositions coincide and recover the standard bias--variance decomposition in Eq.~\eqref{eq:BV_decomposition}.
\end{remark}

From Proposition \ref{prop:BV_decomposition_risk}, we define two types of Bregman unbiasedness.

\begin{definition}[Bregman unbiasedness] \label{def:Bregman_unbiasedness}
\ 
\begin{enumerate}
\item An estimator $\delta(X^{n})$ of $\theta$ is a \textit{type-I Bregman unbiased estimator} if for all $\theta\in\Theta$,
\begin{align}
D_{\varphi}(\theta, \hat{\theta}^{*}_{\delta}) = 0. \label{eq:type_I_unbiased}
\end{align}

\item An estimator $\delta(X^{n})$ of $\theta$ is a \textit{type-II Bregman unbiased estimator} if for all $\theta\in\Theta$,
\begin{align}
D_{\varphi}(\vE_{\theta}[\delta(X^{n})], \theta) = 0. \label{eq:type_II_unbiased}
\end{align}
\end{enumerate}
\end{definition}

These notions of unbiasedness are consistent with the generalized notion of unbiasedness introduced by Lehmann \cite{b5ede1a1-5922-3545-a373-ec4871dfdd4e}.

\begin{definition}[\text{\cite{b5ede1a1-5922-3545-a373-ec4871dfdd4e}}]
Let $\ell(\theta, \hat{\theta})$ be a loss function.
An estimator $\delta(X^{n})$ is unbiased with respect to the loss function $\ell(\theta, \hat{\theta})$ if for all $\theta,\theta^{\prime}\in\Theta$,
\begin{align}
\vE_{\theta}[\ell(\theta, \delta(X^{n}))]
\leq
\vE_{\theta}[\ell(\theta^{\prime}, \delta(X^{n}))].
\end{align}
\end{definition}

\begin{prop} \label{prop:consistency}
\ 
\begin{enumerate}
\item An estimator $\delta(X^{n})$ is unbiased with respect to the loss function $\ell_{\varphi}(\theta,\hat{\theta})$ if and only if $\delta(X^{n})$ is a type-I Bregman unbiased estimator.
\item An estimator $\delta(X^{n})$ is unbiased with respect to the loss function $\tilde{\ell}_{\varphi}(\theta,\hat{\theta})$ if and only if $\delta(X^{n})$ is a type-II Bregman unbiased estimator.
\end{enumerate}
\end{prop}

\begin{proof}
For each of the two losses, apply the corresponding decomposition in Proposition \ref{prop:BV_decomposition_risk} to 
$\vE_{\theta}[\ell_{\varphi}(\theta^{\prime},\delta(X^n))]$ and $\vE_{\theta}[\tilde{\ell}_{\varphi}(\theta^{\prime},\delta(X^n))]$. 
The second term in the decomposition is independent of $\theta^{\prime}$, and hence Lehmann's unbiasedness condition is equivalent to the vanishing of the corresponding bias term. 
\end{proof}

These Bregman unbiasednesses can be characterized as follows.

\begin{prop} \label{prop:characterization_Bregman_unbiasedness}
\ 
\begin{enumerate}
\item An estimator $\delta(X^{n})$ of $\theta$ is a type-I Bregman unbiased estimator if and only if for all $\theta\in\Theta$,
\begin{align}
\nabla \varphi(\theta)
=
\vE_{\theta}\left[\nabla \varphi(\delta(X^{n}))\right]. \label{eq:type_I_unbiased_characterization}
\end{align}

\item An estimator $\delta(X^{n})$ of $\theta$ is a type-II Bregman unbiased estimator if and only if $\delta(X^{n})$ is an unbiased estimator of $\theta$.
\end{enumerate}
\end{prop}
\begin{proof}
For $1)$, the statement follows from Eqs.~\eqref{eq:characterization_x_star} and \eqref{eq:type_I_unbiased}. 
For $2)$, by $1)$ of Proposition \ref{prop:basic_property_Bregman_divergence}, we have
\[
D_{\varphi}(\vE_{\theta}[\delta(X^n)],\theta)=0
\quad\Longleftrightarrow\quad
\vE_{\theta}[\delta(X^n)] = \theta.
\]
This is exactly the standard unbiasedness condition in Definition \ref{def:unbiased}.
\end{proof}

\begin{remark}
From Eq.~\eqref{eq:type_I_unbiased_characterization}, the type-I Bregman unbiasedness of $\delta(X^{n})$ is equivalent to the unbiasedness of the estimator $\nabla\varphi(\delta(X^{n}))$ for $\nabla\varphi(\theta)$ in the dual space induced by $\nabla\varphi$. Meanwhile, type-II Bregman unbiasedness is equivalent to the standard notion of unbiasedness.
\end{remark}
For convenience, we introduce the following notation.

\begin{definition}
Let $\delta(X^{n})$ be an estimator of $\theta$. 
The parameter $\eta$ and the corresponding estimator $\hat{\eta}(X^{n})$ in the dual space induced by $\nabla\varphi$ are defined as
\begin{align}
\eta&:=\nabla\varphi(\theta), \\ 
\hat{\eta}(X^{n})&:=\nabla\varphi(\delta(X^{n})).
\end{align}
\end{definition}
 

Finally, we define Bregman UMVUEs.

\begin{definition}[Bregman UMVUEs]
\ 
\begin{enumerate}
\item A type-I Bregman unbiased estimator $\delta(X^{n})$ is the \textit{type-I Bregman UMVUE} if for all $\theta\in\Theta$ and for all type-I Bregman unbiased estimators $\delta^{\prime}(X^{n})$,
\begin{align}
\vE_{\theta}\left[D_{\varphi}(\hat{\theta}^{*}_{\delta}, \delta(X^{n}))\right]
\leq
\vE_{\theta}\left[D_{\varphi}(\hat{\theta}^{*}_{\delta^{\prime}}, \delta^{\prime}(X^{n}))\right]. 
\end{align}

\item A type-II Bregman unbiased estimator $\delta(X^{n})$ is the \textit{type-II Bregman UMVUE} if for all $\theta\in\Theta$ and for all type-II Bregman unbiased estimators $\delta^{\prime}(X^{n})$,
\begin{align}
&\vE_{\theta}\left[D_{\varphi}(\delta(X^{n}), \vE_{\theta}[\delta(X^{n})])\right]
\notag\\
&\leq
\vE_{\theta}\left[D_{\varphi}(\delta^{\prime}(X^{n}), \vE_{\theta}[\delta^{\prime}(X^{n})])\right].
\end{align}
\end{enumerate}
\end{definition}

From Proposition \ref{prop:characterization_Bregman_unbiasedness}, type-II Bregman unbiasedness is equivalent to the standard notion of unbiasedness.
Moreover, by $2)$ of Proposition \ref{prop:basic_property_Bregman_divergence}, the map
$\hat{\theta}\mapsto \tilde{\ell}_{\varphi}(\theta,\hat{\theta})=D_{\varphi}(\hat{\theta},\theta)$ 
is strictly convex in $\hat{\theta}$.
Thus, by the Rao--Blackwell theorem (Proposition \ref{prop:RB_thm}) and the Lehmann--Scheff{\'e} theorem (Proposition \ref{prop:LS_thm}), the estimator constructed from a complete sufficient statistic is also optimal in the type-II Bregman sense.

\begin{prop} \label{prop:typeII_UMVUE}
Let $T$ be a complete sufficient statistic for $\mathcal{P}$. 
If there exists an unbiased estimator of $\theta$ that is a function of $T$, then it is also the type-II Bregman UMVUE for any generator function $\varphi$ that satisfies the assumptions in Section \ref{ssec:Bregman_divergence}.
\end{prop}

\begin{remark}
Proposition~\ref{prop:typeII_UMVUE} indicates that the type-II Bregman UMVUE does not yield a genuinely new extension of the classical framework.
Therefore, we subsequently focus on the type-I Bregman UMVUE, which yields a different notion of unbiasedness and optimality.
\end{remark}

\section{Main Results} \label{sec:main_result}

In this section, we derive analogs of the Rao--Blackwell and Lehmann--Scheff{\' e} theorems for type-I Bregman UMVUEs.
We first present an extended Rao--Blackwell theorem. The proof proceeds by applying Rao--Blackwellization in the dual space induced by $\nabla\varphi$ and then mapping back via $(\nabla\varphi)^{-1}$.

\begin{theorem}[Extended Rao--Blackwell Theorem] \label{thm:Bregman_Rao_Blackwell}
Let $T$ be a sufficient statistic for $\mathcal{P}$ and $\delta(X^{n})$ be an estimator with a finite risk. 
Define
\begin{align}
\hat{\eta}(X^{n}) &:= \nabla\varphi(\delta(X^{n})), \\
\tilde{\eta}(T) &:= \vE[\hat{\eta}(X^{n}) \mid T], \\
\tilde{\delta}(T) &:= (\nabla\varphi)^{-1}(\tilde{\eta}(T)).
\end{align}
Then,
\begin{align}
R_{\ell_{\varphi}}(\theta, \tilde{\delta}) \leq R_{\ell_{\varphi}}(\theta, \delta).
\end{align}
In particular, if $\delta(X^{n})$ is type-I Bregman unbiased, then so is $\tilde{\delta}(T)$, and
\begin{align}
\vE_{\theta}\!\left[D_{\varphi}(\hat{\theta}^{*}_{\tilde{\delta}}, \tilde{\delta}(T))\right]
\leq
\vE_{\theta}\!\left[D_{\varphi}(\hat{\theta}^{*}_{\delta}, \delta(X^{n}))\right].
\end{align}
\end{theorem}

\begin{proof} Let $\eta:=\nabla\varphi(\theta)$. Then, 
\begin{align}
\ell_{\varphi}(\theta, \delta(X^{n})) &= D_{\varphi}(\theta, \delta(X^{n})) \\
&\overset{(a)}{=} D_{\varphi^{*}}(\nabla\varphi(\delta(X^{n})), \nabla\varphi(\theta)) \\ 
&= D_{\varphi^{*}}(\hat{\eta}(X^{n}), \eta), 
\end{align}
where $(a)$ follows from $3)$ of Proposition \ref{prop:basic_property_Bregman_divergence}.
By the convexity of $D_{\varphi^{*}}(\cdot, \eta)$, we obtain 
\begin{align}
D_{\varphi^{*}}(\tilde{\eta}(T), \eta) &= D_{\varphi^{*}}(\vE[\hat{\eta}(X^{n})\mid T], \eta) \\ 
&\overset{(b)}{\leq} \vE[D_{\varphi^{*}}(\hat{\eta}(X^{n}), \eta)\mid T], 
\end{align}
where $(b)$ follows from Jensen's inequality. 
Taking the expectations on $T$, by the law of total expectation, we obtain 
\begin{align}
\vE[D_{\varphi^{*}}(\tilde{\eta}(T), \eta)] &\leq \vE[D_{\varphi^{*}}(\hat{\eta}(X^{n}), \eta)]. 
\end{align}
From $3)$ of Proposition \ref{prop:basic_property_Bregman_divergence}, we obtain 
\begin{align}
&\vE[D_{\varphi}(\nabla\varphi^{*}(\eta), \nabla\varphi^{*}(\tilde{\eta}(T)))] \notag \\ 
&\leq \vE[D_{\varphi}(\nabla\varphi^{*}(\eta), \nabla\varphi^{*}(\hat{\eta}(X^{n})))]. 
\end{align}
Using Eq.~\eqref{eq:nabla_formula}, we obtain 
\begin{align}
&\vE[D_{\varphi}((\nabla\varphi)^{-1}(\eta), (\nabla\varphi)^{-1}(\tilde{\eta}(T)))] \notag \\ 
&\leq \vE[D_{\varphi}((\nabla\varphi)^{-1}(\eta), (\nabla\varphi)^{-1}(\hat{\eta}(X^{n})))], 
\end{align}
which leads to 
\begin{align}
R_{\ell_{\varphi}}(\theta, \tilde{\delta}) &= \vE[D_{\varphi}(\theta, \tilde{\delta}(T))] \\ 
&\leq \vE[D_{\varphi}(\theta, \delta(X^{n}))] = R_{\ell_{\varphi}}(\theta, \delta).
\end{align}

From Eqs.~\eqref{eq:characterization_x_star} and \eqref{eq:BV_decomposition_type_I}, to prove the last statement, it suffices to show that $\tilde{\eta}(T)$ is an unbiased estimator of $\eta$. Indeed, 
\begin{align}
\vE[\tilde{\eta}(T)] &=  \vE[\vE[\hat{\eta}(X^{n})\mid T]] \\ 
&= \vE[\hat{\eta}(X^{n})] \overset{(c)}{=} \eta, 
\end{align}
where $(c)$ follows from the fact that $\delta(X^{n})$ is a type-I Bregman unbiased estimator. 
\end{proof}

Next, we derive Lehmann--Scheff{\' e} theorem for type-I Bregman UMVUE. 
\begin{lemma} \label{lem:extended_uniqueness}
Let $T$ be a complete sufficient statistic for $\mathcal{P}$.
Let $\tilde{\eta}_{1}(T)$ and $\tilde{\eta}_{2}(T)$ be unbiased estimators for $\eta:=\nabla\varphi(\theta)$, and define
\begin{align}
\tilde{\delta}_{i}(T):=(\nabla\varphi)^{-1}(\tilde{\eta}_{i}(T)), \quad i=1,2.
\end{align}
Then,
\begin{align}
\tilde{\delta}_{1}(T) = \tilde{\delta}_{2}(T), \quad \text{a.e. } \mathcal{P}.
\end{align}
\end{lemma}
\begin{proof}
From Lemma \ref{lem:uniqueness}, we obtain $\tilde{\eta}_{1}(T)=\tilde{\eta}_{2}(T), \text{a.e. $\mathcal{P}$}$. Thus, 
$\tilde{\delta}_{1}(T) = (\nabla\varphi)^{-1}(\tilde{\eta}_{1}(T)) = (\nabla\varphi)^{-1}(\tilde{\eta}_{2}(T)) = \tilde{\delta}_{2}(T), \text{a.e. $\mathcal{P}$}$.
\end{proof}

\begin{theorem}[Extended Lehmann--Scheff{\' e} Theorem] \label{thm:Bregman_Lehmann_Scheffe}
Let $T$ be a complete sufficient statistic for $\mathcal{P}$.
Assume that there exists a type-I Bregman unbiased estimator $\delta(X^{n})$ of $\theta$, and define
\begin{align}
\hat{\eta}(X^{n}) &:= \nabla\varphi(\delta(X^{n})), \\ 
\tilde{\eta}(T)&:=\vE[\hat{\eta}(X^{n})\mid T], \\ 
\tilde{\delta}(T)&:= (\nabla\varphi)^{-1}(\tilde{\eta}(T)). 
\end{align}
Then, $\tilde{\delta}(T)$ is the unique type-I Bregman UMVUE of $\theta$.
\end{theorem}
\begin{proof}
Let $\delta^{\prime}(X^{n})$ be any type-I Bregman unbiased estimator of $\theta$ and define 
\begin{align}
\hat{\eta}^{\prime}(X^{n}) &:= \nabla\varphi(\delta^{\prime}(X^{n})), \\ 
\tilde{\eta}^{\prime}(T) &:= \vE[\hat{\eta}^{\prime}(X^{n})\mid T], \\ 
\tilde{\delta}^{\prime}(T) &:= (\nabla\varphi)^{-1}(\tilde{\eta}^{\prime}(T)).
\end{align}
Then, we obtain 
\begin{align}
\vE_{\theta}\left[D_{\varphi}(\hat{\theta}^{*}_{\tilde{\delta}}, \tilde{\delta}(T))\right] 
&\overset{(a)}{=} \vE_{\theta}\left[D_{\varphi}(\hat{\theta}^{*}_{\tilde{\delta}^{\prime}}, \tilde{\delta}^{\prime}(T))\right] \\
&\overset{(b)}{\leq} \vE_{\theta}\left[D_{\varphi}(\hat{\theta}^{*}_{\delta^{\prime}}, \delta^{\prime}(X^{n}))\right], 
\end{align}
where $(a)$ follows from Lemma \ref{lem:extended_uniqueness}; $(b)$ follows from Theorem \ref{thm:Bregman_Rao_Blackwell}.
\end{proof}

The following result follows from Lemma \ref{lem:extended_uniqueness} and Theorem \ref{thm:Bregman_Lehmann_Scheffe}.
\begin{cor} \label{cor:Bregman_Lehmann_Scheffe}
Let $T$ be a complete sufficient statistic and $\tilde{\eta}(T)$ be an unbiased estimator of $\eta:=\nabla\varphi(\theta)$.
Then, 
\begin{align}
\tilde{\delta}(T):=(\nabla\varphi)^{-1}(\tilde{\eta}(T))
\end{align}
is the unique type-I Bregman UMVUE.
\end{cor}

Finally, we present some examples of a type-I Bregman UMVUE of $\theta$.
\begin{eg}
Let $X^{n}$ be a random sample from the exponential distribution $\textsf{Exp}(\theta)$ with $n>1$ and $\varphi(x)=-\log x$.
Then, $\nabla\varphi(x)=-1/x$ and $T=\sum_{i=1}^{n} X_{i}$ is complete sufficient for $\theta$. 
Because $T \sim \textsf{Gamma}(n,1/\theta)$, we obtain $\vE[1/T]=1/(n-1)\theta$ for $n>1$.
Thus, 
\begin{align}
\tilde{\eta}(T)=-\frac{n-1}{T}
\end{align}
is unbiased for $\eta=-1/\theta$.
Therefore, 
\begin{align}
\tilde{\delta}(T)=\frac{T}{n-1} = \frac{1}{n-1}\sum_{i=1}^{n}X_{i}
\end{align}
is the type-I Bregman UMVUE. 
Meanwhile, the classical UMVUE of $\theta$ is
$\bar{X} = \frac{1}{n}\sum_{i=1}^{n}X_{i}$, which differs from the type-I Bregman UMVUE.
\end{eg}

\begin{eg}
Let $X^{n}$ be a random sample from the log-normal distribution with unknown (reparametrized) location parameter and known scale parameter $\sigma^{2}$, denoted by $\textsf{LogNormal}(\log\theta, \sigma^{2})$, and let $\varphi(x):= x\log x - x$. 
Then, $\nabla\varphi(x)=\log x$ and $T=\sum_{i=1}^{n} \log X_{i}$ is a complete sufficient statistic for $\log\theta$.
Because $\vE[T]=n\log\theta$,$\tilde{\eta}(T)={T}/{n}$ 
is unbiased for $\eta=\log\theta$.
Thus,
\begin{align}
\tilde{\delta}(T)=\exp\!\left(\frac{1}{n}\sum_{i=1}^{n} \log X_i\right) \quad (\text{geometric mean of $X^{n}$})
\end{align}
is the type-I Bregman UMVUE.
Meanwhile, the classical UMVUE of $\theta$ is $\exp\left\{\frac{1}{n}\sum_{i=1}^{n}\log X_{i} - \frac{\sigma^{2}}{2n}\right\}$, 
which is different from the type-I Bregman UMVUE.
\end{eg}

\begin{eg}
Let $X^{n}$ be a random sample from the Bernoulli distribution with natural
parameter $\theta\in\mathbb{R}$:
\begin{align}
p_{X\mid\theta}(x)
&=
\exp\{\theta x-A(\theta)\},
\end{align}
where $A(\theta)=\log(1+e^{\theta})$. 
Let $\varphi(\theta):=A(\theta)$. Then,
$p:=\nabla\varphi(\theta)={e^{\theta}}/{(1+e^{\theta})}$ is the expectation parameter.
Since $T:=\sum_{i=1}^{n}X_{i}$ is
complete sufficient statistic for $p$ and $\vE_{\theta}\left[{T}/{n}\right]=p$, 
$T/n$ is an unbiased estimator of the dual parameter $p$.
Hence, allowing the extended natural parameter space
$\overline{\mathbb{R}}:=\mathbb{R}\cup\{-\infty,+\infty\}$, the corresponding
type-I Bregman UMVUE is
\begin{align}
\tilde{\delta}(T)
&=\log\frac{T}{n-T} = \log \frac{\sum_{i=1}^{n}X_{i}}{n-\sum_{i=1}^{n}X_{i}},
\end{align}
where $\tilde{\delta}(0)=-\infty$ and
$\tilde{\delta}(n)=+\infty$ by convention.
In contrast, the natural parameter $\theta=\log p/(1-p)$ has no classical unbiased estimator for finite $n$.
Indeed, a function of $p$ is unbiasedly estimable from $n$ Bernoulli trials
only if it is a polynomial in $p$ of degree at most $n$
\cite[Chapter~2, Example~3.1]{LC}. Thus, a classical UMVUE of $\theta$ does not exist.
\end{eg}

\section{Conclusion} \label{sec:conclusion}
In this paper, we developed a theory of UMVUEs under Bregman losses, which we termed \emph{Bregman UMVUEs}. 
By exploiting bias--variance-type decompositions for Bregman divergences, we introduced two notions of unbiasedness, referred to as type-I and type-II Bregman unbiasedness, corresponding to the two natural loss functions $D_{\varphi}(\theta,\hat{\theta})$ and $D_{\varphi}(\hat{\theta},\theta)$.
We showed that the type-II formulation reduces to the classical setting: the induced notion of unbiasedness coincides with the standard definition. 
Meanwhile, the type-I formulation results in a fundamentally different framework, in which unbiasedness is characterized through the dual space induced by the gradient map $\nabla\varphi$. 
This induces a new class of estimators that, in general, differs from the classical UMVUE.
For the type-I formulation, we established analogs of the Rao--Blackwell and Lehmann--Scheff{\'e} theorems by leveraging the dual space representation associated with Bregman divergences. 
These results provide a systematic method for constructing type-I Bregman UMVUEs from sufficient and complete sufficient statistics. 
Several examples demonstrated that the resulting estimators can differ substantially from their classical counterparts.

\section*{Acknowledgment}
\textit{Disclosure of AI usage:}
The proofs of Theorems 1 and 2 were developed with assistance from OpenAI's ChatGPT-5.4 through iterative dialogue on proof strategies and intermediate arguments. The authors independently verified the resulting proofs and take full responsibility for their correctness.


\end{document}